\documentclass[10pt,a4paper,oneside]{article}

\usepackage[T1]{fontenc}
\usepackage[latin1]{inputenc}
\usepackage[english]{babel}
\usepackage{fancyhdr}
\usepackage{amsmath}
\usepackage{amsthm}
\usepackage{amsfonts}
\usepackage{amssymb}
\usepackage{amscd}
\usepackage{amsbsy}
\usepackage{lastpage}
\usepackage{enumitem}
\usepackage{mathrsfs}% script fonts
\usepackage[sort&compress,longnamesfirst]{natbib}
\usepackage{comment}
\usepackage{color}
\usepackage[bookmarks=true,colorlinks=true,linkcolor=UniBlue,citecolor=UniRed,filecolor=UniBlue,urlcolor=UniBlue,pdftitle={Continuous-Time Portfolio Optimisation for a Behavioural Investor with Bounded Utility on Gains},pdfauthor={Mikl\'{o}s R\'{a}sonyi and Andrea Meireles Rodrigues},pdfcreator={Andrea Meireles Rodrigues},%
pdfsubject={This paper examines an optimal investment problem in a continuous-time (essentially) complete financial market with a finite horizon. We deal with an investor who behaves consistently with principles of Cumulative Prospect Theory, and whose utility function on gains is bounded above. The well-posedness of the optimisation problem is trivial, and a necessary condition for the existence of an optimal trading strategy is derived. This condition requires that the investor's probability distortion function on losses does not tend to $0$ near $0$ faster than a given rate, which is determined by the utility function. Under additional assumptions, we show that this condition is indeed the borderline for attainability, in the sense that for slower convergence of the distortion function there does exist an optimal portfolio.},%
pdfkeywords={Behavioural\ finance;\ } {Bounded\ utility;\ } {Choquet\ integral;\ } {Con\-tin\-u\-ous-time\ models;\ } {Market\ completeness;\ } {Non-concave\ utility;\ } {Optimal\ portfolio;\ } {Probability\ distortion.}]{hyperref}
\usepackage{cleveref}% last package

\crefname{sec}{section}{sections}
\crefname{def}{definition}{definitions}
\crefname{th}{theorem}{theorems}
\crefname{lem}{lemma}{lemmata}
\crefname{prop}{proposition}{propositions}
\crefname{cor}{corollary}{corollaries}
\crefname{ex}{example}{examples}
\crefname{as}{assumption}{assumptions}
\crefname{obs}{remark}{remarks}
\crefname{app}{appendix}{appendices}

\definecolor{UniBlue}{RGB}{0,50,95}
\definecolor{UniRed}{RGB}{193,0,67}

\bibpunct{[}{]}{;}{n}{,}{,}

\newcommand{\bbOne}{1\hspace*{-0.8ex}1}

\numberwithin{equation}{section} %section

\newtheorem{theorem}{Theorem}[section]%section
\newtheorem{proposition}[theorem]{Proposition}
\newtheorem{lemma}[theorem]{Lemma}
\newtheorem{corollary}[theorem]{Corollary}
\newtheorem{assumption}[theorem]{Assumption}

\theoremstyle{definition}
\newtheorem{definition}[theorem]{Definition}

\theoremstyle{remark}
\newtheorem{remark}[theorem]{Remark}

\newtheoremstyle{example}% name of the style to be used
  {0.5\topsep}% measure of space to leave above the theorem. E.g.: 3pt
  {0.5\topsep}% measure of space to leave below the theorem. E.g.: 3pt
  {\itshape}% name of font to use in the body of the theorem
  {0pt}% measure of space to indent
  {\normalfont}% name of head font
  {.}% punctuation between head and body
  {5pt plus 1pt minus 1pt}% space after theorem head; " " = normal interword space
  {}% Manually specify head
\theoremstyle{example}
\newtheorem{example}[theorem]{Example}

\title{Continuous-Time Portfolio Optimisation\\for a Behavioural Investor with Bounded Utility on Gains%
\thanks{The authors wish to thank an anonymous referee for a careful reading of their manuscript and valuable suggestions. The first version of this paper was finalised during the workshop ``Modeling Market Dynamics and Equilibrium'' at the Hausdorff Institute, Bonn, in August 2013. M.~R\'{a}sonyi thanks the organisers for their kind invitation and the Institute for its hospitality. A.\,M.~Rodrigues gratefully acknowledges the financial support of FCT - Funda\c{c}\~{a}o para a Ci\^{e}ncia e Tecnologia (Portuguese Foundation for Science and Technology) through the Doctoral Grant SFRH/BD/69360/2010.}}

\date{1st September 2013}

\author{Miklós~Rásonyi\footnote{MTA Alfr\'{e}d R\'{e}nyi Institute of Mathematics, Budapest, Hungary, and School of Mathematics, University of Edinburgh, Edinburgh, Scotland, U.K..}\and Andrea~M.~Rodrigues\footnote{Corresponding author. School of Mathematics, University of Edinburgh, Edinburgh, Scotland, U.K..\newline\hspace*{1.8em}E-mail addresses:~\texttt{\href{mailto:Miklos.Rasonyi@ed.ac.uk}{Miklos.Rasonyi@ed.ac.uk}} and \texttt{\href{mailto:A.S.Meireles-Rodrigues@sms.ed.ac.uk}{A.S.Meireles-Rodrigues@sms.ed.ac.uk}}.}}

\begin{document}

\fancypagestyle{plain}{ %
  \fancyhf{} % remove everything
  \renewcommand{\headrulewidth}{0pt} % remove lines as well
  \renewcommand{\footrulewidth}{0pt}
}
\thispagestyle{plain}

\maketitle

\begin{abstract}
	This paper examines an optimal investment problem in a continuous-time (essentially) complete financial market with a finite horizon. We deal with an investor who behaves consistently with principles of Cumulative Prospect Theory, and whose utility function on gains is bounded above. The well-posedness of the optimisation problem is trivial, and a necessary condition for the existence of an optimal trading strategy is derived. This condition requires that the investor's probability distortion function on losses does not tend to $0$ near $0$ faster than a given rate, which is determined by the utility function. Under additional assumptions, we show that this condition is indeed the borderline for attainability, in the sense that for slower convergence of the distortion function there does exist an optimal portfolio.
	
\end{abstract}

%%%%%%%%%%%%%%%%%%%%%%%%%%%%%%%%%%%%%%%%%%%%%%%%%%%%%%%%%%%%%%%%%%%
%%                                                               %%
%% Body of article       																				 %%
%%                                                               %%
%%%%%%%%%%%%%%%%%%%%%%%%%%%%%%%%%%%%%%%%%%%%%%%%%%%%%%%%%%%%%%%%%%%

\pagestyle{fancy}
\fancyhf{} % remove everything
\renewcommand{\headrulewidth}{0pt} % remove lines as well
\renewcommand{\footrulewidth}{0pt}
\cfoot{Page \thepage/\pageref{LastPage}}
\chead{{\small Portfolio Optimisation for a Behavioural Investor with Bounded Utility on Gains}}

%=========================================================================================================================================
%Introduction ============================================================================================================================
%=========================================================================================================================================
\section{Introduction and Summary}\label[sec]{sec:Intro}

The optimal investment problem is a classical one in financial mathematics, and it has been widely studied in the framework of Expected Utility Theory (EUT, for short), formulated by \citet{neumann53}. This theory presumes that any rational investor's preferences can be numerically represented by a so-called utility function, usually assumed concave and increasing.

Over the years, as some of EUT's fundamental principles have been questioned by empirical studies, several alternative theories have emerged, amongst which the Cumulative Prospect Theory (CPT) proposed by \citet{kahneman79} and \citet{tversky92}. Within this framework, the utility function, which is still assumed to be strictly increasing with wealth, is no longer globally concave. This is because investors, whilst generally risk averse on gains, were found to become risk seeking when undergoing losses. The existence of a reference point defining gains and losses is also presumed, a feature that is absent in EUT. Lastly, according to CPT, economic agents find it hard to assess probabilities rationally and objectively. Instead, they are subjective and systematically miscalculate probabilities (for example, events of small probability tend to be overweighted), which is modelled with functions distorting the probability measure.

As a consequence, the behavioural agent's objective functional to be maximised involves a nonlinear Choquet integral. This raises new, mathematically complex challenges, and the most common approaches to solving the EUT portfolio problem, such as dynamic programming or the use of convex duality methods, are not suitable anymore.

The literature on cumulative prospect theory in continuous-time models is scarce. \citet{berkelaar04}, \citet{carlier-dana2011}, and \citet{reichlin2012} consider utilities defined on the positive real axis (and we remark further that, in the first paper, no probability distortions are considered, which considerably simplifies the problem). The only studies about the whole real line case are \citet{jin2008} and \citet{rasonyi2013}. \citet{jin2008} find explicit solutions in certain cases, but under hypotheses (see Assumption~4.1 therein) which are neither easily verifiable nor economically interpretable. Existence of optimisers for the case of power-like distortion and utility functions has been shown in \citet{rasonyi2013}, with necessary and sufficient conditions on the parameters. However, the case of utilities growing slower than a power function remained open. We address this problem in the present paper, in the setting of bounded above utility functions.

As it is widely stated in the literature, the paper by \citet{menger34} (whose English translation can be found in \citep{menger67}) appears to have been the first to assert the necessity of a boundedness assumption on the utility function in order to avoid a St.~Petersburg-type paradox. Even though this has lead to a considerable amount of debate, several authors have since advocated and made further arguments for considering bounded utilities (see e.g.\ \citet{arrow70,arrow74,arrow51}, \citet{markowitz76}, and \citet{savage54}, to cite only a few). We refer to \citet{muravievrogers13}, who provide a strong argument against unbounded utilities (which they attribute to Kenneth Arrow). Thus, in this paper we restrict ourselves to the case where the utility is bounded above. As \Cref{obs:NecI} below shows, we cannot impose that the utility is bounded below, as this would contradict the existence of an optimiser.

In \Cref{sec:SetUp}, the model is presented, the principles of CPT are formalised, and the optimisation problem is rigorously stated. \Cref{sec:WPAttain} deals with the issues of well-posedness and existence, \Cref{sec:Conclusion} concludes. For the sake of a simple exposition, all auxiliary results and proofs are compiled in \Cref{app:Proofs}.

%=========================================================================================================================================
%Notation and Set-Up =====================================================================================================================
%=========================================================================================================================================
\section{Notation and Set-Up}\label[sec]{sec:SetUp}

%The Market ==============================================================================================================================
\subsection{The Market}

Let us consider a continuous-time and frictionless financial market with trading interval $\left[0,T\right]$, where $T\in\left(0,+\infty\right)$ is a fixed nonrandom horizon. As usual, we start with a complete probability space $\left(\Omega,\mathscr{F},\mathbb{P}\right)$. We suppose further that the evolution of information through time is modelled by a filtration, $\mathbb{F}=\left\{\mathscr{F}_{t}\text{; }0\leq t\leq T\right\}$, satisfying the \emph{usual conditions} of right-continuity and saturatedness. Finally, we assume for convenience that the $\sigma$-algebra $\mathscr{F}_{0}$ is $\mathbb{P}$-trivial, and also that $\mathscr{F}=\mathscr{F}_{T}$.

Next, we fix an arbitrary $d\in\mathbb{N}$, and introduce a $d$-dimensional c\`adl\`ag, adapted process $S=\left\{S_{t}\text{; }0\leq t\leq T\right\}$. For each $i\in\left\{1,\ldots,d\right\}$, $S_{t}^{i}$ represents the price of a certain risky asset $i$ at time $t$. In addition to these $d$ risky {securities}, we shall assume that the market contains a riskless asset %
%, sometimes called \emph{money market}, with price process $S^{0}=\left\{S_{t}^{0}\text{; }0\leq t\leq T\right\}$. For simplicity of notation, and without loss of generality, we presume the money market to be constant, that is, 
$S_{t}^{0}\equiv 1$ for any $t\in\left[0,T\right]$. % (economically speaking, this means that the \emph{interest rate} is zero). 
Therefore, we shall work directly with discounted prices. Let us make the following technical assumptions throughout.

\begin{assumption}\label[as]{as:continuousCDFrho}
	There exists a measure $\mathbb{Q}$ on $\left(\Omega,\mathscr{F}\right)$, equivalent to $\mathbb{P}$ (we write $\mathbb{Q}\sim\mathbb{P}$), such that the (discounted) price process $S$ is a $\mathbb{Q}$-local martingale.\footnote{In particular, $S$ is a semi-martingale.} Furthermore, setting $\rho\triangleq d\mathbb{Q}/d\mathbb{P}$ (the Radon-Nikodym derivative of $\mathbb{Q}$ with respect to $\mathbb{P}$), the cumulative distribution function (CDF) of $\rho$ under $\mathbb{P}$, denoted by $F_{\rho}^{\mathbb{P}}$, is continuous.\footnote{We recall that the \emph{cumulative distribution function} of $\rho$, with respect to the probability measure $\mathbb{P}$, is given by $F_{\rho}^{\mathbb{P}}\!\left(x\right)=\mathbb{P}\!\left(\rho \leq x\right)$, for every real number $x$. We note further that $F_{\rho}^{\mathbb{Q}}$ is also continuous by $\mathbb{Q}\sim\mathbb{P}$.}
\end{assumption}

\begin{assumption}\label[as]{as:esssup}
	The essential supremum of $\rho$ with respect to $\mathbb{P}$, $\mathrm{ess\,sup}_{\mathbb{P}}\,\rho$, is infinite.
\end{assumption}

\begin{assumption}\label[as]{as:rhoinW}
	Both $\rho$ and $1/\rho$ belong to $\mathscr{W}$, where $\mathscr{W}$ is defined as the family of all real-valued random variables $Y$ satisfying $\mathbb{E}_{\mathbb{P}}\!\left[\left|Y\right|^{p}\right]<+\infty$ for all $p>0$.
\end{assumption}

We recall that a \emph{portfolio} (or \emph{trading strategy}) over the time interval $\left[0,T\right]$ is an 
$S$-integrable, $\mathbb{R}^{d}$-valued stochastic process $\left\{\phi_{t}\text{; }0\leq t\leq T\right\}$. For every $i\in\left\{1,\ldots,d\right\}$, $\phi_{t}^{i}$ represents the position in the $i$-th asset at time $t$. We assume that trading is self-financing, so the (discounted) value $\Pi_t^{\phi}$ of the portfolio at $t$, for all $t\in\left[0,T\right]$, is given by $\Pi_{t}^{\phi}=x_0+\int_{0}^{t}\phi_{s}\,dS_{s}$, where $x_0$ is the investor's initial capital. The set of portfolios is denoted by $\Phi\!\left(x_{0}\right)$.

In order to preclude arbitrage opportunities, we must restrict ourselves to a subset $\Psi\!\left(x_{0}\right)\subseteq \Phi\!\left(x_{0}\right)$ of \emph{admissible} strategies. Amongst several possible admissibility criteria, one which is often adopted in the literature is that the portfolio's wealth process should be uniformly bounded below by some constant (possibly depending on the portfolio). However, in the present paper, as in \citet{rasonyi2013} and for the reasons given therein, we assume that admissible strategies are those whose (discounted) wealth process is a martingale under $\mathbb{Q}$ (and not only a local martingale).

Finally, we fix a scalar-valued random variable $B$ satisfying $\mathbb{E}_{\mathbb{Q}}\!\left[\left|B\right|\right]<+\infty$, representing a benchmark. Hereafter, we shall also assume, essentially, that the market is complete.
\begin{assumption}\label[as]{as:kindcompl}
	The random variable $B$ and all $\sigma\!\left(\rho\right)$-measurable random variables in $L^{1}\!\left(\mathbb{Q}\right)$ (i.e., integrable with respect to the measure $\mathbb{Q}$) are replicable, that is, each of them is equal to the terminal value of some admissible portfolio $\phi\in\Psi\!\left(x_{0}\right)$.
\end{assumption}

%The Investor ============================================================================================================================
\subsection{The Investor}

We consider a small \emph{CPT investor} with a given initial capital $x_{0}\in\mathbb{R}$.

Firstly, the agent is assumed to have a \emph{reference point}, represented by the replicable claim $B$ introduced above, with respect to which payoffs are evaluated. Thus, given a payoff $X$ at the terminal time $T$ and a scenario $\omega\in\Omega$, the investor is said to make a \emph{gain} (respectively, a \emph{loss}) if $X\!\left(\omega\right)>B\!\left(\omega\right)$ (respectively, $X\!\left(\omega\right)<B\!\left(\omega\right)$).

Secondly, the agent's preferences towards risk are described by a \emph{non-concave utility function} $u:\mathbb{R}\rightarrow\mathbb{R}$, given by%
\begin{equation}
	u\!\left(x\right)\triangleq u_{+}\!\left(x^{+}\right)\bbOne_{\left[\left.0,+\infty\right)\right.}\!\left(x\right)-u_{-}\!\left(x^{-}\right)\bbOne_{\left(-\infty,0\right)}\!\left(x\right),\qquad x\in\mathbb{R},%
	\footnote{%
		Here, $x^{+}\triangleq \max\!\left\{x,0\right\}$ and $x^{-}\triangleq -\min\!\left\{x,0\right\}$, for any real number $x$.}
\end{equation}
where the strictly increasing, continuous functions $u_{\pm}:\left[\left.0,+\infty\right)\right.\rightarrow \left[\left.0,+\infty\right)\right.$, satisfy $u_{\pm}\!\left(0\right)=0$. %
% and $u_{-}\!\left(1\right)=1$.
%, are respectively the investor's \emph{utility on gains} and \emph{on losses}.
Note that no assumptions are made concerning the differentiability or the concavity of the functions. Moreover, it is clear that the functions $u_{\pm}$ have (possibly infinite) limits as $x\rightarrow+\infty$. In what follows, the notation $u_{\pm}\!\left(+\infty\right)\triangleq \lim_{x\rightarrow+\infty} u_{\pm}\!\left(x\right)$ will be used.
\begin{assumption}[\textbf{Bounded utility on gains}]\label[as]{as:boundedup}
	The utility on gains is bound\-ed above, i.e., $M\triangleq u_{+}\!\left(+\infty\right)<+\infty$.
\end{assumption}

\begin{example}
	\begin{enumerate}[label=\emph{(\roman*)}]
		\item
		The \emph{exponential utility} with parameter $\alpha>0$ is the function %$u:\left[\left.0,+\infty\right)\right.\rightarrow \left[\left.0,+\infty\right)\right.$
		given by $u\!\left(x\right)\triangleq 1-e^{-\alpha x}$ for all $x\geq 0$.
		
		\item
		The \emph{power utility} with parameter $\alpha\in \mathbb{R}\setminus\left\{0\right\}$ is the function $u:\left[\left.0,+\infty\right)\right.\rightarrow \left[\left.0,+\infty\right)\right.$ defined as $u\!\left(x\right)\triangleq x^{\alpha}$ for $\alpha> 0$ and $u\!\left(x\right)\triangleq 1-(1+x)^{\alpha}$ for $\alpha<0$. It is trivial that $u$ is bounded above if and only if $\alpha<0$.
	
		\item
		The \emph{logarithmic utility} is the function defined by $u\!\left(x\right)\triangleq \log\!\left(1+x\right)$ for every $x\geq 0$.
	\end{enumerate}
\end{example}

The third and most prominent feature of CPT is that the investor has a distorted perception of the actual probabilities, which is modelled by the two strictly increasing, continuous \emph{probability distortion}
functions $w_{\pm}:\left[0,1\right]\rightarrow\left[0,1\right]$ (\emph{on gains} and \emph{on losses}, respectively), with $w_{\pm}\!\left(0\right)=0$ and $w_{\pm}\!\left(1\right)=1$. The economic agent is said to \emph{overweight} (respectively, \emph{underweight}) small-prob\-a\-bil\-i\-ty losses if, for all $x$ in some right-neighbourhood of zero, we have $w_{-}\!\left(x\right)\geq x$ (respectively, $w_{-}\!\left(x\right)\leq x$). An entirely analogous definition can be given for small-probability gains.%, as well as large probabilities.

\begin{example}
	\begin{enumerate}[label=\emph{(\roman*)}]
		\item
		The \emph{power distortion} with parameter $\beta>0$ is the function given by $w\!\left(x\right)\triangleq x^{\beta}$ for every $x\in\left[0,1\right]$.
		
		\item
		The distortion defined as $w\!\left(x\right)\triangleq \exp\!\left\{-\beta\left[-\log\!\left(x\right)\right]^{\varpi}\right\}\bbOne_{\left.\left(0,1\right.\right]}\!\left(x\right)$ for all $x\in\left[0,1\right]$, with parameters $\varpi\in\left(0,1\right)$ and $\beta>0$, was first proposed by \citet{prelec98}.
	\end{enumerate}
\end{example}

%The Optimal Portfolio Problem ===========================================================================================================
\subsection{The Optimal Investment Problem}

The continuous-time portfolio selection problem for a behavioural investor with CPT preferences consists of choosing an optimal investment strategy, that is, one that maximises a certain expected distorted payoff functional. 
\begin{definition}[\textbf{Behavioural optimal investment problem}]
	The math\-e\-mat\-i\-cal formulation of the \emph{behavioural optimal portfolio problem} is:
	\begin{equation}\label{eq:optport}
		\textrm{maximise } \left\{V\!\left(\Pi^{\phi}_{T}-B\right)=V_{+}\!\left(\left[\Pi^{\phi}_{T}-B\right]^{+}\right)-V_{-}\!\left(\left[\Pi^{\phi}_{T}-B\right]^{-}\right)\right\}\footnote{Note that $V\!\left(\Pi^{\phi}_{T}-B\right)$ may well be $-\infty$ for certain $\phi\in\Psi\!\left(x_{0}\right)$.}
	\end{equation}
	over $\phi \in \Psi\!\left(x_{0}\right)$, where
	\begin{equation}
		V_{\pm}\!\left(\left[\Pi^{\phi}_{T}-B\right]^{\pm}\right)\triangleq \int_{0}^{+\infty}w_{\pm}\!\left(\mathbb{P}\!\left\{u_{\pm}\!\left(\left[\Pi^{\phi}_{T}-B\right]^{\pm}\right)>y\right\}\right)\,dy.
	\end{equation}
%	are the \emph{Choquet integrals} of $u_{\pm}\!\left(\left[\Pi^{\phi}_{T}-B\right]^{\pm}\right)$ with respect to the non-additive set functions $w_{\pm}\circ\mathbb{P}$.\footnote{The symbol $\circ$ denotes function composition.} 
Setting $V^{*}\!\left(x_{0}\right)\triangleq \sup\left\{V\!\left(\Pi^{\phi}_{T}-B\right)\text{: }\phi \in \Psi\!\left(x_{0}\right)\right\}$, we say that $\phi^{*}\in\Psi\!\left(x_{0}\right)$ is an \emph{optimal strategy} if $V\!\left(\Pi^{\phi^{*}}_{T}-B\right)=V^{*}\!\left(x_{0}\right)$. 
\end{definition}

\begin{remark}\label[obs]{obs:nearopt}
	One may wonder why the existence of an optimal $\phi^{*}$ is relevant when the existence of $\varepsilon$-optimal strategies $\phi^{\varepsilon}$ (i.e., ones that are $\varepsilon$-close to the supremum over all strategies) is automatic, for all $\varepsilon>0$. There are at least two, closely related reasons for this.

	Firstly, non-existence of an optimal $\phi^{*}$ usually means that an optimiser sequence $\left\{\phi^{1/n}\text{; }n\in\mathbb{N}\right\}$ shows wild, extreme behaviour (e.g., they converge to infinity, see Example~7.3 of \citet{rs05}). Such strategies are both practically infeasible and economically counter-intuitive.

	Secondly, existence of $\phi^{*}$ normally goes together with some compactness property (tightness of laws in the present paper). Such a property seems necessary for the convergence of any potential numerical procedure to find an optimal (or at least an $\varepsilon$-optimal) strategy.
\end{remark}

Henceforward, we shall assume for simplicity that $B=0$. We may do this without loss of generality since $B$ is replicable by \Cref{as:kindcompl}.

%=========================================================================================================================================
%Well-Posedness and Attainability ========================================================================================================
%=========================================================================================================================================
\section{Well-Posedness and Attainability}\label[sec]{sec:WPAttain}

Well-posedness is trivial in our current setting.
%Intuitively speaking, if $V^{*}=+\infty$, then the investor can obtain an arbitrarily high degree of satisfaction from the trading strategies that are available in the market, so maximisation does not make sense. In order to exclude this possibility, it is common in the literature to presume beforehand that the supremum in \eqref{eq:optport} is finite. The first result of the present paper (whose proof is absolutely trivial, hence omitted) allows us to conclude that, regardless of what the utility on losses and what the probability distortions are, when the utility on gains is bounded above, the maximisation problem~\eqref{eq:optport} is always \emph{well-posed} in this sense.
\begin{proposition}\label[prop]{prop:boundeduplus} Under \Cref{as:boundedup},
$V^{*}\!\left(x_{0}\right)\leq u_{+}\!\left(+\infty\right)<+\infty$.\qed
\end{proposition}

It may still be the case that an optimal solution does not exist. We must now study whether or not this finite supremum $V^{*}\!\left(x_{0}\right)$ is indeed a maximum, that is, whether or not the optimisation problem is \emph{attainable}. A first and important answer is given by the following result.
\begin{theorem}[\textbf{Necessary condition I}]\label[th]{th:NecI}
	Under \Cref{as:continuousCDFrho,as:kindcompl,as:esssup,as:boundedup}, there exists an optimal portfolio for problem~\eqref{eq:optport} only if
	\begin{equation}
		\liminf_{x\rightarrow 0^{+}} w_{-}\!\left(x\right)u_{-}\!\left(\frac{1}{x}\right)>0.
	\end{equation}
\end{theorem}

\begin{remark}
	\begin{enumerate}[label=\emph{(\roman*)}]
		\item
		In particular, \Cref{th:NecI} implies that, if $u_{-}\!\left(+\infty\right)<+\infty$ as well, then the optimisation problem is not attainable. Although many authors  argue in favour of such $u_{-}$, see e.g.\ \citet{muravievrogers13}, for the remainder of this section, we shall only consider the case where $u_{-}$ is not bounded.
		
		\item
		Considering the specific case where both $u_{-}$ and $w_{-}$ are power functions, respectively with parameters $\alpha>0$ and $\beta>0$, there is an optimal strategy only if $\alpha\geq\beta$, so we obtain the analogue of Proposition~3.7 in \citet{rasonyi2013}. Moreover, trivial modifications in the proof of \Cref{th:NecI} show that, when $\alpha=\beta$, and thus $\lim_{x\rightarrow 0^{+}} w_{-}\!\left(x\right)u_{-}\!\left(1/x\right)=1$, existence of an optimal portfolio still does not hold.
		
		\item
		Another interesting conclusion which can be drawn from the above result is that, under additional conditions on the growth of $u_{-}$,\footnote{E.g., there exist $\gamma\in\left[\left.0,1\right)\right.$, $C_{1}>0$ and $C_{2}\geq 0$ such that
$u_{-}\!\left(x\right)\leq C_{1}\,x^{\gamma}+C_{2}$ for sufficiently large $x$.} the investor must distort the probability of losses, otherwise there is no optimal portfolio. This complements Theorem~3.2 of \citet{jin2008} (which states that a probability distortion on losses is a necessary condition for the well-posedness of~\eqref{eq:optport} when $u_{+}\!\left(+\infty\right)=+\infty$), but for a bounded utility on gains.%\rhombus
	\end{enumerate}
	\label[obs]{obs:NecI}
\end{remark}

For example, an investor with a logarithmic utility and a Prelec distortion on losses does not admit an optimal trading strategy. Existence of an optimal strategy requires that $w_{-}\!\left(x\right)$ cannot decrease to zero too fast, but must approach zero more slowly than $\left[u_{-}\!\left(1/x\right)\right]^{-1}$, as $x\rightarrow 0^{+}$. Motivated by \Cref{th:NecI}, we introduce the following concept.
\begin{definition}[\textbf{Associated distortion}]\label[def]{def:AssocDist}
	Given a real number $\delta>0$ and a utility function $u_-:\left[\left.0,+\infty\right)\right.\rightarrow \left[\left.0,+\infty\right)\right.$ with $u_-\!\left(+\infty\right)=+\infty$, let us define the function $w_{\delta}:\left[0,1\right]\rightarrow \left[0,1\right]$ in the following way,
	\begin{equation}\label{eq:wdelta}
		%w^{\delta}\!\left(x\right)\triangleq \left\{
		%	\begin{array}{ll}
		%		0,& \text{if } x=0,\\
		%		\left[u\!\left(1/x\right)\right]^{-\delta},& \text{if } x\in\left.\left(0,1\right.\right].
		%	\end{array}
		%\right.
		w_{\delta}\!\left(x\right)\triangleq u_-^{\delta}(1)\left[u_-\!\left(1/x\right)\right]^{-\delta}\bbOne_{\left.\left(0,1\right.\right]}\!\left(x\right),\qquad x\in\left[0,1\right].
	\end{equation}
	We call $w_{\delta}$ the \emph{distortion associated with $u_-$ with parameter $\delta$}.
\end{definition}

\begin{example}\label[ex]{ex:Prelec}
	Let $\alpha>0$ and $\varpi\in\left(0,1\right)$, and consider $u_{-}:\left[\left.0,+\infty\right)\right.\rightarrow \left[\left.0,+\infty\right)\right.$ given by $u_{-}\!\left(x\right)\triangleq\exp\left\{\alpha\,\mathrm{sgn}\!\left(x-1\right)\left|\log\!\left(x\right)\right|^{\varpi}\right\}\bbOne_{\left(0,+\infty\right)}\!\left(x\right)$ for any $x\geq 0$. %
	%\footnote{We recall that the \emph{signum function}, $\mathrm{sgn}:\mathbb{R}\rightarrow\left\{-1,0,1\right\}$, is defined as follows,
	%	\begin{equation*}
	%		\mathrm{sgn}\!\left(x\right)\triangleq \left\{%
	%			\begin{array}{ll}
	%				-1,& \text{if } x<0,\\
	%				0,& \text{if } x=0,\\
	%				1,& \text{if } x>0.
	%		\end{array}
	%		\right.
	%	\end{equation*}} %
	Clearly, this utility function satisfies $u_{-}\!\left(+\infty\right)=+\infty$ and, for every $\delta>0$, its associated distortion is the Prelec distortion with parameters $\delta\alpha>0$ and $\varpi\in\left(0,1\right)$.
\end{example}

The following corollary to \Cref{th:NecI} is now immediate and tells us that, in the particular case where the distortion on losses is the distortion associated with $u_{-}$ for some parameter $\delta>0$, a necessary condition for attainability is that $\delta\leq1$.
\begin{corollary}[\textbf{Necessary condition II}]\label[cor]{cor:NecII}
	Let $u_{-}\!\left(+\infty\right)=+\infty$ and $\delta>0$. Suppose that the investor's probability weighting on losses is $w_-=w_{\delta}$. Then, under \Cref{as:continuousCDFrho,as:esssup,as:kindcompl,as:boundedup}, the optimal portfolio problem~\eqref{eq:optport} is attainable only if $\delta\leq1$.\qed
\end{corollary}
Therefore, when the parameter $\delta$ is strictly greater than $1$, by the preceding result we know that the supremum in~\eqref{eq:optport} is never attained. The same conclusion also holds with $\delta=1$ for some fairly typical utility functions (see \Cref{obs:NecI} above).

The remainder of this section will be devoted to arguing that the condition $\delta<1$ is not only ``almost necessary'', but also sufficient to ensure that an optimal trading strategy does in fact exist, under an additional hypothesis on $u_{-}$ below.
\begin{assumption}\label[as]{as:existsxi}
	For every $\delta\in\left(0,1\right)$, there is some $\xi>1$ such that
	\begin{equation}
		\lim_{x\rightarrow +\infty} \frac{\left[u_{-}\!\left(x^{\xi}\right)\right]^{\delta}}{u_{-}\!\left(x\right)}=0.
	\end{equation}
\end{assumption}
As an almost reciprocal of \Cref{cor:NecII}, we have the following.
\begin{theorem}[\textbf{Sufficient condition}]\label[th]{th:Suff}
	Suppose $u_{-}$ and $w_{\delta}$ are as in the statement of \Cref{cor:NecII}, and $w_{-}\!\left(x\right)\geq 
w_{\delta}\!\left(x\right)$ for all $x\in [0,1]$. Under \Cref{as:boundedup,as:existsxi,as:continuousCDFrho,as:kindcompl,as:rhoinW}, if $\delta\in\left(0,1\right)$, then there exists an optimal strategy.
\end{theorem}
Hence, \Cref{cor:NecII} and \Cref{th:Suff} show that $\left[u_{-}\!\left(1/x\right)\right]^{-1}$ can be regarded as the threshold for the distortion function as far as the existence of an optimal portfolio is concerned. Below this, in the sense of $\delta<1$, attainability holds. Above this, when $\delta>1$ (or, for some cases, also  when $\delta=1$), it does not. Finally, we present a result which allows us to associate \Cref{as:existsxi} to the renowned concept of \emph{asymptotic elasticity} (first introduced in the financial mathematics literature by \citet{cvitanic96} and \citet{kramkov99}).
\begin{lemma}\label[lem]{lem:AEzfinite}
	Suppose $u_{-}\!\left(+\infty\right)=+\infty$, and let $z_{-}:\left[\left.0,+\infty\right)\right.\rightarrow \left[\left.0,+\infty\right)\right.$ be the transform of $u_{-}$ given by $z_{-}\!\left(x\right)\triangleq \log\!\left(u_{-}\!\left(e^{x}\right)\right)$, for all $x\geq0$. %
	%\begin{equation}\label{eq:transfz}
	%	z_{-}\!\left(x\right)\triangleq \log\!\left(u_{-}\!\left(e^{x}\right)\right),\qquad x\geq 0.
	%\end{equation}
	%
	If there exist $\gamma>0$ and $\underline{x}> 0$ such that
	\begin{equation}\label{eq:AEz}
		z_{-}\!\left(\lambda x\right)\leq \lambda^{\gamma} z_{-}\!\left(x\right)\qquad \text{for all } \lambda\geq1 \text{ and } x\geq \underline{x},
	\end{equation}
	then \Cref{as:existsxi} is satisfied.
\end{lemma}

\begin{remark}\label[obs]{obs:AEz}
Suppose further that the function $z_{-}$ is continuously differentiable on $\left(x_{0},+\infty\right)$, for some $x_{0}\geq 0$. It can be easily verified that, in this case, condition~\eqref{eq:AEz} is equivalent to
	\begin{equation*}
		AE_{+}\!\left(z_{-}\right)\triangleq \limsup_{x\rightarrow+\infty}\frac{x\,\left(z_{-}\right)'\!\left(x\right)}{z_{-}\!\left(x\right)}<+\infty,
	\end{equation*}
where $AE_{+}\!\left(z_{-}\right)$ is the asymptotic elasticity of $z_{-}$ at $+\infty$. We refer to Lemma~6.3 in \citet{kramkov99}, while drawing attention to the fact that the proof there only uses the continuity, the monotonicity and the continuous differentiability of $z_{-}$, not its concavity.%\rhombus
\end{remark}

\begin{example}
	\begin{enumerate}[label=\emph{(\roman*)}]
		\item
		Suppose $u_{-}$ is continuously differentiable and $AE_{+}\!\left(u_{-}\right)<+\infty$. If, in addition, there exist constants $C>0$, $\gamma>0$ so that $u_{-}\!\left(x\right)\geq C\,x^{\gamma}$ holds true for all $x$ sufficiently large, then $u_{-}$ satisfies \Cref{as:existsxi}. Indeed,
		\begin{equation*}
			\frac{x\,\left(z_{-}\right)'\!\left(x\right)}{z_{-}\!\left(x\right)}\leq \frac{x\,\left(z_{-}\right)'\!\left(x\right)}{\log\!\left(C\right)+\gamma x}=\frac{\left(z_{-}\right)'\!\left(x\right)}{\left(\log\!\left(C\right)/x\right)+\gamma}
		\end{equation*}
		for every sufficiently large $x$, thus $AE_{+}\!\left(z_{-}\right)\leq \frac{1}{\gamma}\limsup_{x\rightarrow+\infty} \left(z_{-}\right)'\!\left(x\right)$. But, as noted in \citet[p.~946]{kramkov99}, it is trivial to check that $\limsup_{x\rightarrow+\infty} \left(z_{-}\right)'\!\left(x\right)=AE_{+}\!\left(u_{-}\right)$, which is finite by hypothesis, hence \Cref{lem:AEzfinite} gives us the claimed result.
		
		In particular, this implies that the power utility function with parameter $\alpha>0$ (not necessarily less than one), having asymptotic elasticity equal to $\alpha$, verifies \Cref{as:existsxi}.
		
		\item
		Let $u_{1}$ be the utility of \Cref{ex:Prelec} with parameters $\alpha>0$ and $\varpi\in\left(0,1\right)$, $u_{2}$ the logarithmic utility, and $u_{3}$ the \emph{log-log utility} defined as $u_{3}\!\left(x\right)\triangleq \log\!\left(1+\log\!\left(1+x\right)\right)$ for all $x\geq 0$. Their transforms, $z_{1}$, $z_{2}$ and $z_{3}$, respectively, equal
		\begin{align*}
			z_{1}\!\left(x\right) &=\alpha\,x^{\varpi},\\
			z_{2}\!\left(x\right) &=\log\!\left(\log\!\left(1+e^{x}\right)\right),\\
			z_{3}\!\left(x\right) &=\log\!\left(\log\!\left(1+\log\!\left(1+e^{x}\right)\right)\right),
		\end{align*}
		for all $x\geq 0$.	It can be checked that these functions are strictly concave, hence $AE_{+}\!\left(z_{i}\right)\leq 1$ for all $i\in\left\{1,2,3\right\}$ (see, e.g., \citet[Lem\-ma~6.1]{kramkov99}).

		\item
		Assume $u_{-}\!\left(+\infty\right)=+\infty$, and also that $\left(u_{-}\right)'$ exists and tends to $0$ fast enough as $x\rightarrow+\infty$, i.e., $\left(u_{-}\right)'\!\left(x\right)\leq C/\left[x\,\log\!\left(x\right)\right]$ for some $C>0$ and for $x$ large enough. Then \Cref{as:existsxi} is fulfilled. Indeed,
		\begin{equation*}
			\frac{x\,\left(z_{-}\right)'\!\left(x\right)}{z_{-}\!\left(x\right)}=\frac{x\,e^{x}\left(u_{-}\right)'\!\left(e^{x}\right)}{u_{-}\!\left(e^{x}\right)\log\!\left(u_{-}\!\left(e^{x}\right)\right)}\leq\frac{C}{u_{-}\!\left(e^{x}\right)\log\!\left(u_{-}\!\left(e^{x}\right)\right)}\xrightarrow[x\rightarrow +\infty]{} 0.
		\end{equation*}
\end{enumerate}
\end{example}

%=========================================================================================================================================
%Conclusion ==============================================================================================================================
%=========================================================================================================================================
\section{Conclusions and Further Work}\label[sec]{sec:Conclusion}

In this work, we analysed the CPT optimal portfolio problem in a continuous-time complete financial market. We focused solely on the case where the investor's utility on gains is bounded above and we found a necessary condition for the existence of an optimal solution. As expected, the obtained condition involves both the utility and the distortion on losses, whereas gains do not matter. A sufficient condition for attainability was derived too, showing that our necessary condition forms the threshold for existence.% Yet we remark that, as in \citet{rasonyi2013}, no explicit form for the optimal terminal wealth was obtained.

With regard to our \Cref{as:existsxi}, which may appear to be somewhat artificial at first, it was shown to be related to such widely known a concept as asymptotic elasticity. Moreover, it is satisfied by a large class of functions, including some of the most popular ones in the literature. Extending these results for unbounded $u_{+}$ is the object of further research.% An immediate remark is that similar arguments to those in the proof of \Cref{th:NecI} give that $\lim_{x\rightarrow 0^{+}}w_{-}\!\left(x\right)u_{-}\!\left(1/x\right)=+\infty$ and $\limsup_{x\rightarrow 0^{+}}w_{+}\!\left(x\right)u_{+}\!\left(1/x\right)<+\infty$ are necessary conditions for the problem to be well-posed in this case.

%=========================================================================================================================================
%Appendix: Proofs and Supporting Results =================================================================================================
%=========================================================================================================================================
\appendix
\section{Proofs and Auxiliary Results}\label[app]{app:Proofs}

We may and will assume that $u_{-}\!\left(1\right)=1$. Indeed, let $y>0$ be the (unique) value such that $u_{-}\!\left(y\right)=1$. Define $\bar{u}_{\pm}\!\left(x\right)\triangleq u_{\pm}\!\left(xy\right)$. Notice that \Cref{as:boundedup,as:existsxi} continue to hold for $\bar{u}_{\pm}$ and that $V^{*}_{u_{-}}\!\left(x_{0}\right)=V^{*}_{\bar{u}_{-}}(x_{0}/y)$, so all the results below extend from 
the case $u_{-}\!\left(1\right)=1$ to the general case.

\begin{lemma}\label[lem]{lem:Supequpinfty}
	Under \Cref{as:boundedup}, there exists an optimal portfolio for problem~\eqref{eq:optport} only if
	\begin{equation}\label{eq:supequpinfty}
		\sup\left\{V\!\left(\Pi_{T}^{\phi}\right)\text{: }\phi \in \Psi\!\left(x_{0}\right)\right\} <u_{+}\!\left(+\infty\right).
	\end{equation}
\end{lemma}

\begin{proof}
Omitted.
\end{proof}

%=========================================================================================================================================
\begin{proof}[Proof of \Cref{th:NecI}]
	The proof is by contraposition. Let us suppose that we have $\liminf_{x\rightarrow 0^{+}} w_{-}\!\left(x\right) u_{-}\!\left(1/x\right)=0$. %
	Then, using \Cref{as:continuousCDFrho,as:esssup}, it is possible to find two sequences of strictly positive real numbers $\left\{a_{n}\text{; }n\in\mathbb{N}\right\}$ and $\left\{b_{n}\text{; }n\in\mathbb{N}\right\}$, respectively strictly decreasing and strictly increasing, with $\lim_{n\rightarrow +\infty} a_{n}=0$ and $\lim_{n\rightarrow +\infty} b_{n}=+\infty$, whose terms satisfy both $\mathbb{P}\!\left\{\rho\leq b_{n}\right\}=1-a_{n}$ and $w_{-}\!\left(a_{n}\right)u_{-}\!\left(1/a_{n}\right)<1/n$.
		
	Now, for every $n\in\mathbb{N}$, we define the event $A_{n}\triangleq \left\{\rho\leq b_{n}\right\}$, as well as the positive and $\sigma\!\left(\rho\right)$-measurable random variable $X_{n}\triangleq \frac{b_{n}}{2\,\mathbb{Q}\!\left(A_{n}\right)}\bbOne_{A_{n}}$.	It is straightforward to see that $\lim_{n\rightarrow+\infty}\mathbb{Q}\!\left(A_{n}\right)=\lim_{n\rightarrow+\infty}\mathbb{P}\!\left(A_{n}\right)=1$ and so
	\begin{equation*}
		V_{+}\!\left(X_{n}\right)=u_{+}\!\left(\frac{b_{n}}{2\,\mathbb{Q}\!\left(A_{n}\right)}\right)w_{+}\!\left(\mathbb{P}\!\left(A_{n}\right)\right)\xrightarrow[n\rightarrow +\infty]{} u_{+}\!\left(+\infty\right).
	\end{equation*}
	
	Next, let $Y_{n}\triangleq \frac{b_{n}-2x_{0}}{2\,\mathbb{Q}\!\left(A_{n}^{c}\right)}\bbOne_{A_{n}^{c}}$ (note that $\mathbb{Q}\!\left(A_{n}^{c}\right)>0$ for all $n\in\mathbb{N}$), which is also $\sigma\!\left(\rho\right)$-measurable. Since $\lim_{n\rightarrow+\infty}b_{n}=+\infty$, there is an integer $n_{0}$ such that $b_{n}>2x_{0}$ for any $n\geq n_{0}$. Furthermore, given that $\lim_{n\rightarrow+\infty} \frac{b_{n}-2x_{0}}{2b_{n}}=1/2$, there must be some $n_{1}\in\mathbb{N}$ so that $\frac{b_{n}-2x_{0}}{2b_{n}}<1$ for all $n\geq n_{1}$. Combining these facts %
	with the inequality $\mathbb{Q}\!\left(A_{n}^{c}\right)=\mathbb{E}_{\mathbb{Q}}\!\left[\rho\bbOne_{A_{n}^{c}}\right]\geq b_{n}\,\mathbb{P}\!\left(A_{n}^{c}\right)$ and %
	with the monotonicity of $u_{-}$ yields, for every $n\geq \max\left\{n_{0},n_{1}\right\}$,
	\begin{align*}
		V_{-}\!\left({Y_{n}}\right)= u_{-}\!\left(\frac{b_{n}-2x_{0}}{2\,\mathbb{Q}\!\left(A_{n}^{c}\right)}\right)w_{-}\!\left(\mathbb{P}\!\left(A_{n}^{c}\right)\right)
		\leq u_{-}\!\left(\frac{1}{\mathbb{P}\!\left(A_{n}^{c}\right)}\right)w_{-}\!\left(\mathbb{P}\!\left(A_{n}^{c}\right)\right)<\frac{1}{n}.
	\end{align*}
	%for every $n\geq \max\left\{n_{0},n_{1}\right\}$, where the first inequality follows trivially from $\mathbb{Q}\!\left(A_{n}^{c}\right)=\mathbb{E}_{\mathbb{Q}}\!\left[\rho\bbOne_{A_{n}^{c}}\right]\geq b_{n}\,\mathbb{P}\!\left(A_{n}^{c}\right)$.
	
	Hence, setting $Z_{n}=X_{n}-Y_{n}$, $n\in\mathbb{N}$, it is obvious that $Z_{n}$ is $\sigma\!\left(\rho\right)$-measurable, and also that $\mathbb{E}_{\mathbb{Q}}\!\left[Z_{n}\right]=x_{0}$ by construction. Besides, for every $n\geq n_{0}$, we have $V_{-}\!\left(Z_{n}^{-}\right)=V_{-}\!\left(Y_{n}\right)<+\infty$ and $\mathbb{E}_{\mathbb{Q}}\!\left[\left|Z_{n}\right|\right]=b_{n}-x_{0}<+\infty$, therefore $Z_{n}$ is replicable from initial capital $x_0$. Finally, we get that $\liminf_{n\rightarrow +\infty}V\!\left(Z_{n}\right)\geq u_{+}\!\left(+\infty\right)-0$, %
	%\begin{equation*}
	%	\liminf_{n\rightarrow +\infty}V\!\left(Z_{n}\right)\geq u_{+}\!\left(+\infty\right)-0,
	%\end{equation*}
	so by \Cref{lem:Supequpinfty} we can conclude.
\end{proof}

%=========================================================================================================================================
\begin{lemma}\label[lem]{lem:zetaG}
	The following three statements are equivalent,
	\begin{enumerate}[label=\emph{(\roman*)}]
		\item
		\Cref{as:existsxi} holds true,
		
		\item
		For each $\delta\in\left(0,1\right)$, there exist a real number $\zeta>1$ and a decreasing function $G:\left(0,+\infty\right)\rightarrow\left[\left.1,+\infty\right)\right.$ such that, for every $\lambda>0$,
		\begin{equation}\label{eq:zetaG}
			u_{-}\!\left(x^{\zeta}\right)\leq \left[\lambda\,u_{-}\!\left(x\right)\right]^{1/\delta},
		\end{equation}
		for all $x\geq G\!\left(\lambda\right)$, and
		
		\item
		For every $\delta\in\left(0,1\right)$, there is $\varsigma>1$ such that $\lim_{x\rightarrow+\infty} \left[z_{-}\!\left(x\right)-\delta\,z_{-}\!\left(\varsigma x\right)\right]=+\infty$, %
		%\begin{equation}
		%	\lim_{x\rightarrow+\infty} \left[z_{-}\!\left(x\right)-\delta\,z_{-}\!\left(\varsigma x\right)\right]=+\infty,
		%\end{equation}
		where $z_{-}$ is the transform of $u_{-}$ defined in \Cref{lem:AEzfinite}.
	\end{enumerate}
\end{lemma}

\begin{proof}
	\emph{(i)}$\,\Rightarrow\,$\emph{(ii)} is trivial, so we prove the reverse implication. Let $\delta\in\left(0,1\right)$ be fixed, and consider $\lambda>0$ arbitrary. Since, by hypothesis, $\lim_{x\rightarrow +\infty} \frac{\left[u_{-}\!\left(x^{\xi}\right)\right]^{\delta}}{u_{-}\!\left(x\right)}=0$, there exists some $L\triangleq L\!\left(\lambda\right)\geq 1$ such that $u_{-}\!\left(x^{\xi}\right)<\left[\lambda\,u_{-}\!\left(x\right)\right]^{1/\delta}$ for all $x\geq L$. Next define, for each $\lambda>0$, the nonempty set
	\begin{equation*}
		\mathscr{S}_{\lambda}\triangleq \left\{L\geq 1\text{: } u_{-}\!\left(x^{\xi}\right)<\left[\lambda\,u_{-}\!\left(x\right)\right]^{1/\delta} \text{ for all } x\geq L\right\},
	\end{equation*}
	which is bounded below by $1$, so it admits an infimum. Then let $G:\left(0,+\infty\right)\rightarrow \mathbb{R}$ be the function given by $G\!\left(\lambda\right)\triangleq \inf \mathscr{S}_{\lambda}$, for any $\lambda>0$. Clearly, by construction, $G\geq 1$. Furthermore, it can be easily checked that, for every $\lambda>0$ and for all $x\geq G\!\left(\lambda\right)$, the inequality $u_{-}\!\left(x^{\xi}\right)\leq \left[\lambda\,u_{-}\!\left(x\right)\right]^{1/\delta}$ holds true. Finally, it remains to show that $G$ is indeed a decreasing function of $\lambda$. To see this, let $0<\lambda_{1}\leq \lambda_{2}$. Then, for all $x\geq G\!\left(\lambda_{1}\right)\geq 1$, we have $u_{-}\!\left(x^{a}\right)\leq \left[\lambda_{1} u_{-}\!\left(x\right)\right]^{1/\delta}\leq \left[\lambda_{2} u_{-}\!\left(x\right)\right]^{1/\delta}$, hence $G\!\left(\lambda_{1}\right)$ belongs to $\mathscr{S}_{\lambda_{2}}$. Consequently, we must have, by the definition of the infimum, that $G\!\left(\lambda_{1}\right)\geq G\!\left(\lambda_{2}\right)$.
	
	The proof of \emph{(i)}$\,\Leftrightarrow\,$\emph{(iii)} is straightforward.
\end{proof}

%=========================================================================================================================================
\begin{lemma}\label[lem]{lem:wPfX}
	Suppose $u_{-}\!\left(+\infty\right)=+\infty$, and let $f:\left[\left.0,+\infty\right)\right.\rightarrow \left[\left.0,+\infty\right)\right.$ be a continuous, strictly increasing function satisfying both $f\!\left(0\right)=0$ and $f\!\left(+\infty\right)=+\infty$. Then
	\begin{equation}\label{eq:wPfX}
		w_{-}\!\left(\mathbb{P}\!\left\{f\!\left(X\right)>t\right\}\right)\leq \frac{1}{u_{-}\!\left(f^{-1}\!\left(t\right)\right)}\int_{0}^{+\infty} w_{-}\!\left(\mathbb{P}\!\left\{u_{-}\!\left(X\right)>y\right\}\right)\,dy
	\end{equation}
	for any $t>0$ and for any positive random variable $X$.
\end{lemma}

\begin{proof}
	The proof is similar to that of Lemma~3.12 in~\citet{rasonyi2013}, with trivial modifications.%
\end{proof}

%=========================================================================================================================================
\begin{corollary}\label[cor]{cor:PXs}
	Suppose $u_{-}\!\left(+\infty\right)=+\infty$, and let $\delta>0$ be arbitrary. If $w_{\delta}$ is the distortion associated with the utility $u_{-}$ (with parameter $\delta$), then for any $s>0$ we have
	\begin{equation}\label{eq:PXs}
		\mathbb{P}\!\left\{X^{s}>t\right\}\leq \left[\left(u_{-}\right)^{-1}\!\left(\left[\frac{u_{-}\!\left(t^{1/s}\right)}{\int_{0}^{+\infty} w_{\delta}\!\left(\mathbb{P}\!\left\{u_{-}\!\left(X\right)>y\right\}\right)\,dy}\right]^{1/\delta}\right)\right]^{-1}
	\end{equation}
	for all $t>0$ and for all positive random variables $X$.\qed
\end{corollary}

%=========================================================================================================================================
\begin{lemma}\label[lem]{lem:EXeta}
	Suppose $u_{-}\!\left(+\infty\right)=+\infty$ and $\delta\in\left(0,1\right)$. Let Assumption \ref{as:existsxi} be satisfied, and let the decreasing function $G:\left(0,+\infty\right)\rightarrow \left[\left.1,+\infty\right)\right.$ and the real number $\zeta>1$ be those given by \Cref{lem:zetaG}.	Then, for every $\eta\in\left(1,\zeta\right)$, there exists a constant $C>0$ such that, for all positive random variables $X$, we have
	\begin{equation}\label{eq:EXeta}
		\mathbb{E}_{\mathbb{P}}\!\left[X^{\eta}\right]\leq C+\frac{\left[G\!\left(\left[V_{\delta}\!\left(X\right)\right]^{-1}\right)\right]^{\eta}}{\left(u_{-}\right)^{-1}\!\left(\left[V_{\delta}\!\left(X\right)\right]^{-1/\delta}\right)},
	\end{equation}
	with $V_{\delta}\!\left(X\right)\triangleq \int_{0}^{+\infty} w_{\delta}\!\left(\mathbb{P}\!\left\{u_{-}\!\left(X\right)>y\right\}\right)\,dy$.
\end{lemma}

\begin{proof}
	Fix $\delta\in\left(0,1\right)$ and $\eta\in\left(1,\zeta\right)$, and let $X$ be a positive random variable. If $X=0$ $\mathbb{P}$-a.s., then $\mathbb{E}_{\mathbb{P}}\!\left[X^{\eta}\right]=0$ and $V_{-}^{\delta}\!\left(X\right)=0$, hence the inequality~\eqref{eq:EXeta} is satisfied trivially for any $C>0$. So suppose now that $\mathbb{P}\!\left\{X>0\right\}>0$, which implies $V_{\delta}\!\left(X\right)>0$. Using \Cref{cor:PXs},
	\begin{equation}\label{eq:aux1}
		\mathbb{E}_{\mathbb{P}}\!\left[X^{\eta}\right]=\int_{0}^{\infty} \mathbb{P}\!\left\{X^{\eta}>t\right\}\,dt\leq 1+\int_{1}^{+\infty} \left[u^{-1}\!\left(\left[\frac{u\!\left(t^{1/\eta}\right)}{V_{\delta}\!\left(X\right)}\right]^{1/\delta}\right)\right]^{-1} dt,
	\end{equation}
	for any positive random variable $X$.

	%Since, by hypothesis, \Cref{as:existsxi} holds true, we may
We apply \Cref{lem:zetaG} to obtain, %
% a real number $\zeta>1$ and a decreasing function $G:\left(0,+\infty\right)\rightarrow \left[\left.1,+\infty\right)\right.$ such that, for every $\lambda>0$, the inequality~\eqref{eq:zetaG} is verified for any $x\geq G\!\left(\lambda\right)$. Thus, in particular, we have, 
for all $x\geq G\!\left(1/V_{\delta}\!\left(X\right)\right)$,
	\begin{equation*}
		\left(u_{-}\right)^{-1}\left(\left[\frac{u_{-}\!\left(x\right)}{V_{\delta}\!\left(X\right)}\right]^{1/\delta}\right)\geq x^{\zeta},
	\end{equation*}
	where we have also made use of the fact that $\left(u_{-}\right)^{-1}$ is strictly increasing. On the other hand, it follows again from the monotonicity of both $u_{-}$ and $\left(u_{-}\right)^{-1}$ that
	\begin{equation*}
		\left(u_{-}\right)^{-1}\!\left(\left[\frac{u_{-}\!\left(t^{1/\eta}\right)}{V_{\delta}\!\left(X\right)}\right]^{1/\delta}\right)\geq \left(u_{-}\right)^{-1}\!\left(\left[\frac{1}{V_{\delta}\!\left(X\right)}\right]^{1/\delta}\right)
	\end{equation*}
	for all $t\geq 1$. Thus, the preceding facts and the change of variables $x=t^{1/\eta}$ yield
	\begin{align}\label{eq:aux2}
		\lefteqn{\int_{1}^{+\infty} \left[u^{-1}\!\left(\left[\frac{u\!\left(t^{1/\eta}\right)}{V_{\delta}\!\left(X\right)}\right]^{1/\delta}\right)\right]^{-1} dt}\nonumber\\
		&\leq \int_{1}^{\left[G\!\left(1/V_{\delta}\!\left(X\right)\right)\right]^{\eta}} \left[\left(u_{-}\right)^{-1}\!\left(\left[\frac{1}{V_{\delta}\!\left(X\right)}\right]^{1/\delta}\right)\right]^{-1}dt\nonumber\\
		&\quad \hspace{5cm} +\eta\int_{G\!\left(1/V_{\delta}\!\left(X\right)\right)}^{+\infty} \frac{\left[u^{-1}\!\left(\left[\frac{u\!\left(x\right)}{V_{\delta}\!\left(X\right)}\right]^{1/\delta}\right)\right]^{-1}}{x^{1-\eta}}\,dx\nonumber\\
		&\leq \frac{\left[G\!\left(\left[V_{\delta}\!\left(X\right)\right]^{-1}\right)\right]^{\eta}-1}{\left(u_{-}\right)^{-1}\!\left(\left[V_{\delta}\!\left(X\right)\right]^{-1/\delta}\right)}+\eta \int_{1}^{+\infty} \frac{1}{x^{1+\zeta-\eta}}\,dx,
	\end{align}
	and we note that the second integral is finite because $\zeta-\eta>0$.
	
	Hence, plugging~\eqref{eq:aux2} into \eqref{eq:aux1}, setting $C\triangleq 1+\eta \int_{1}^{+\infty} \frac{1}{x^{1+\zeta-\eta}}\,dx\in\left(1,+\infty\right)$ and noting that $\left[G\!\left(\left[V_{\delta}\!\left(X\right)\right]^{-1}\right)\right]^{\eta}-1\leq \left[G\!\left(\left[V_{\delta}\!\left(X\right)\right]^{-1}\right)\right]^{\eta}$ allows us to finally deduce the claimed inequality.
\end{proof}

%=========================================================================================================================================
\begin{proof}[Proof of \Cref{th:Suff}]
	Essentially, we shall follow the proof of The\-o\-rem~4.7 in~\citet{rasonyi2013}, while borrowing some key ideas from \citet{reichlin2012}.
	
	%Let $\delta\in\left(0,1\right)$ be arbitrary, but fixed. 
We begin by taking a maximising sequence $\left\{\phi^{(n)}\text{; }n\in\mathbb{N}\right\}\subseteq \Psi\!\left(x_{0}\right)$, that is, a sequence of admissible trading strategies $\phi^{(n)}$ such that
	\begin{equation*}
		\lim_{n\rightarrow+\infty} V\!\left(\Pi^{\phi^{(n)}}_{T}\right)=V^{*}\!\left(x_{0}\right).
	\end{equation*}
	We shall henceforth denote by $X_{n}$ the terminal wealth of the $n$-th portfolio $\phi^{(n)}$. We clearly have $\inf_{n\in\mathbb{N}} V\!\left(X_{n}\right)>-\infty$. Moreover, we get $\sup_{n\in\mathbb{N}} V_{+}\!\left(X_{n}^{+}\right)<+\infty$ from \Cref{prop:boundeduplus}, hence also
	\begin{equation*}
		\sup_{n\in\mathbb{N}} V_{-}\!\left(X_{n}^{-}\right)\leq \sup_{n\in\mathbb{N}} V_{+}\!\left(X_{n}^{+}\right)-\inf_{n\in\mathbb{N}} V\!\left(X_{n}\right)<+\infty.
	\end{equation*}
	Noting that $w_{-}\geq w_{\delta}$ implies
	\begin{equation*}
		\frac{\left[G\!\left(\left[V_{\delta}\!\left(X_{n}\right)\right]^{-1}\right)\right]^{\eta}}{\left(u_{-}\right)^{-1}\!\left(\left[V_{\delta}\!\left(X_{n}\right)\right]^{-1/\delta}\right)}\leq
\frac{\left[G\!\left(\left[V_{-}\!\left(X_{n}\right)\right]^{-1}\right)\right]^{\eta}}{\left(u_{-}\right)^{-1}\!\left(\left[V_{-}\!\left(X_{n}\right)\right]^{-1/\delta}\right)},
	\end{equation*}
for every $n\in\mathbb{N}$, it then follows from \Cref{lem:EXeta} that $\sup_{n\in\mathbb{N}} \mathbb{E}_{\mathbb{P}}\!\left[\left(X_{n}^{-}\right)^{\eta}\right]<+\infty$, for some $\eta>1$.
	
	Next, $\mathbb{E}_{\mathbb{Q}}\!\left[X_{n}^{+}\right]=x_{0}+\mathbb{E}_{\mathbb{Q}}\left[X_{n}^{-}\right]$,
\Cref{as:rhoinW} and H\"{o}lder's inequality allow us to obtain that $\sup_{n\in\mathbb{N}} \mathbb{E}_{\mathbb{P}}\!\left[\left|X_{n}\right|^{\tau}\right]<+\infty$ for every $\tau\in\left(0,1\right)$ (see the proof of Theorem~4.7 in~\citet{rasonyi2013} for details). %Indeed, fixing $\tau\in\left(0,1\right)$, we start by using H\"{o}lder's inequality and \Cref{as:rhoinW} to derive
%	\begin{align*}
%		\mathbb{E}_{\mathbb{P}}\!\left[\left(X_{n}^{+}\right)^{\tau}\right] &\leq C_{1}\mathbb{E}_{\mathbb{Q}}\!\left[X_{n}^{+}\right]^{\tau}\\
%		&=C_{1}\left(x_{0}+\mathbb{E}_{\mathbb{Q}}\!\left[X_{n}^{-}\right]\right)^{\tau} \\
%		&\leq C_{2}+C_{1}\mathbb{E}_{\mathbb{Q}}\!\left[X_{n}^{-}\right]^{\tau}\\
%		&\leq C_{2}+C_{3}\mathbb{E}_{\mathbb{Q}}\!\left[\left(X_{n}^{-}\right)^{\eta}\right]^{\tau/\eta},
%	\end{align*}
%	and so we have
%	\begin{align*}
%		\mathbb{E}_{\mathbb{P}}\!\left[\left|X_{n}\right|^{\tau}\right] &\leq \mathbb{E}_{\mathbb{P}}\!\left[\left(X_{n}^{+}\right)^{\tau}\right]+\mathbb{E}_{\mathbb{P}}\!\left[\left(X_{n}^{-}\right)^{\tau}\right]\\
%		&\leq C_{2}+C_{3}\mathbb{E}_{\mathbb{Q}}\!\left[\left(X_{n}^{-}\right)^{\eta}\right]^{\tau/\eta}+C_{4}\mathbb{E}_{\mathbb{Q}}\!\left[\left(X_{n}^{-}\right)^{\eta}\right]^{\tau/\eta},
%	\end{align*}
%	where the second inequality is again a trivial consequence of H\"{o}lder's inequality (with $\eta/\tau>1$), hence the claimed result follows from the definition of supremum.
%	
	From this, it is now immediate to conclude that the family $\left\{\mathbb{P}_{X_{n}}\text{; }n\in\mathbb{N}\right\}$, where $\mathbb{P}_{X_{n}}$ denotes the law of the random variable $X_{n}$ with respect to $\mathbb{P}$, is tight. Thus, by Prokhorov's theorem we can extract a weakly convergent subsequence %$\left\{\mathbb{P}_{X_{n_{k}}}\text{; }k\in\mathbb{N}\right\}$, and we write 
$\mathbb{P}_{X_{n_{k}}} \stackrel{w}{\longrightarrow} \nu$ for some probability measure $\nu$.

	Now let $q_{\rho}^{\mathbb{P}}$ denote the quantile function of $\rho$ with respect to $\mathbb{P}$, which is unique up to a set of Lebesgue measure zero.%
	\footnote{We recall that the unique (up to a set of Lebesgue measure zero) \emph{quantile function} of the random variable $\rho$ with respect to the probability measure $\mathbb{P}$, $q_{\rho}^{\mathbb{P}}:\left(0,1\right)\rightarrow \mathbb{R}$, is a generalised inverse of $F_{\rho}^{\mathbb{P}}$, i.e., it is such that
	\begin{equation*}
		 F_{\rho}^{\mathbb{P}}\!\left(q_{\rho}^{\mathbb{P}}\!\left(p\right)-\right)\leq p\leq F_{\rho}^{\mathbb{P}}\!\left(q_{\rho}^{\mathbb{P}}\!\left(p\right)\right)\qquad \text{for any level } p\in\left(0,1\right),
	\end{equation*}
	where $F_{\rho}^{\mathbb{P}}\!\left(x-\right)\triangleq \lim_{s\uparrow x} F_{\rho}^{\mathbb{P}}\!\left(s\right)=\mathbb{P}\!\left\{\rho<x\right\}$. Analogously, given a probability law $\nu$ on the Borel $\sigma$-algebra $\mathscr{B}\!\left(\mathbb{R}\right)$, its quantile function $q_{\nu}$ is the generalised inverse of the distribution function given by $F_{\nu}\!\left(x\right)\triangleq \nu\!\left(\left.\left(-\infty,x\right.\right]\right)$ for any $x\in\mathbb{R}$. The reader is referred to \citet[Appendix~A.3]{follmer04} for a thorough study of quantile functions, their properties and related results.} %
	Then, by our \Cref{as:continuousCDFrho}, the $\sigma\!\left(\rho\right)$-measurable random variable $U\triangleq F_{\rho}^{\mathbb{P}}\!\left(\rho\right)$ follows under $\mathbb{P}$ a uniform distribution on the interval $\left(0,1\right)$, and moreover $\rho=q_{\rho}^{\mathbb{P}}\!\left(U\right)$ $\mathbb{P}$-a.s..
	
	So let us set $X_{*}\triangleq q_{\nu}\!\left(1-U\right)$, which is clearly a $\sigma\!\left(\rho\right)$-measurable random variable. In addition, because $1-U$ is uniformly distributed on $\left(0,1\right)$ under $\mathbb{P}$, we conclude that $X_{*}$ has probability law $\nu$, hence 
%the subsequence of random variables $\left\{X_{n_{k}}\text{; }k\in\mathbb{N}\right\}$ converges in distribution to $X_{*}$ as $k\rightarrow+\infty$, and we write 
$X_{n_{k}}\stackrel{\mathscr{D}}{\longrightarrow} X_{*}$.

	Since $\sup_{n\in\mathbb{N}} V_{\pm}\!\left(X_{n}^{\pm}\right)<+\infty$, it can be shown, exactly as in part~\emph{(i)} of the proof of Theorem~4.7 in \citet{rasonyi2013}, that $V_{\pm}\!\left(X_{*}^{\pm}\right)<+\infty$.

Trivially, we have $0\leq w_{+}\!\left(\mathbb{P}\!\left\{u_{+}\!\left(X_{n_{k}}^{+}\right)>y\right\}\right)\leq \bbOne_{\left[0,M\right]}\!\left(y\right)$ for all $k\in\mathbb{N}$ and for every $y\geq 0$, so 
the Fatou lemma implies $V\!\left(X_{*}\right)\geq V^{*}\!\left(x_0\right)$.

	It remains to check that $\mathbb{E}_{\mathbb{Q}}\!\left[X_{*}\right]\leq x_{0}$. This will be done using an argument of \citet[Proof of Proposition~4.1, p.~16]{reichlin2012}. We remark, however, that some modifications are required to account for the fact that, in our paper, wealth is allowed to become negative.
	
	It is immediate to get that $\mathbb{E}_{\mathbb{Q}}\!\left[X_{*}\right]$ equals
	\begin{equation*}
		\mathbb{E}_{\mathbb{P}}\!\left[\rho\,X_{*}\right]=\mathbb{E}_{\mathbb{P}}\!\left[q_{\rho}^{\mathbb{P}}\!\left(U\right) q_{\nu}\!\left(1-U\right)\right]=\int_{0}^{1} q_{\rho}^{\mathbb{P}}\!\left(x\right) q_{\nu}\!\left(1-x\right)\,dx.
	\end{equation*}
	Furthermore, $q_{\rho}^{\mathbb{P}}$ is positive a.e.\ on $\left(0,1\right)$ because $\rho>0$ a.s., and the fact that the family $\left\{X_{n_{k}}\text{; }k\in\mathbb{N}\right\}$ converges in distribution to $X_{*}$ implies that the sequence of quantile functions $\left\{q^{\mathbb{P}}_{X_{n_{k}}}\text{; }k\in\mathbb{N}\right\}$ converges to $q_{\nu}$ a.e.\ on $\left(0,1\right)$.
	
	Thus, since the positive part function is increasing and continuous, we can combine Fatou's lemma with one of the Hardy-Littlewood inequalities (we refer for instance to \citet[Theorem~A.24]{follmer04}) to obtain
	\begin{align*}
		\int_{0}^{1} q_{\rho}^{\mathbb{P}}\!\left(x\right) \left[q_{\nu}\!\left(1-x\right)\right]^{+}\,dx&\leq \liminf_{k\rightarrow +\infty} \int_{0}^{1} q_{\rho}^{\mathbb{P}}\!\left(x\right) \left[q^{\mathbb{P}}_{X_{n_{k}}}\!\left(1-x\right)\right]^{+}\,dx\\
		&= \liminf_{k\rightarrow +\infty} \int_{0}^{1} q_{\rho}^{\mathbb{P}}\!\left(x\right) q^{\mathbb{P}}_{X^{+}_{n_{k}}}\!\left(1-x\right)\,dx\\
		&\leq \liminf_{k\rightarrow +\infty} \mathbb{E}_{\mathbb{P}}\!\left[\rho\,X_{n_{k}}^{+}\right],
	\end{align*}
	where the equality is a trivial consequence of $\left[q^{\mathbb{P}}_{X_{n_{k}}}\!\left(x\right)\right]^{+}=q^{\mathbb{P}}_{X^{+}_{n_{k}}}\!\left(x\right)$ for a.e.\ $x\in\left(0,1\right)$. On the other hand, it follows from the second Hardy-Littlewood inequality that
	\begin{equation*}
		\mathbb{E}_{\mathbb{P}}\!\left[\rho\,X_{n_{k}}^{-}\right]\leq \int_{0}^{1} q_{\rho}^{\mathbb{P}}\!\left(x\right) q^{\mathbb{P}}_{X^{-}_{n_{k}}}\!\left(x\right)\,dx,%\leq \int_{0}^{1} q_{\rho}^{\mathbb{P}}\!\left(x\right) \left[q_{\nu}\!\left(1-x\right)\right]^{-}\,dx
	\end{equation*}
	for every $k\in\mathbb{N}$.
	
	But the family of a.e.\ positive functions $\left\{q_{\rho}^{\mathbb{P}} q^{\mathbb{P}}_{X^{-}_{n_{k}}}\text{; }k\in\mathbb{N}\right\}$ is uniformly integrable on $\left(0,1\right)$. Indeed, we can choose some $\eta'>1$ such that $\eta'<\eta$, and so H\"{o}lder's inequality with $\eta/\eta'>1$ yields, for all $k\in\mathbb{N}$,
	\begin{align*}
		\int_{0}^{1} \left[q_{\rho}^{\mathbb{P}}\!\left(x\right)q_{X_{n_{k}}^{-}}^{\mathbb{P}}\!\left(x\right)\right]^{\eta'}dx&\leq \mathbb{E}_{\mathbb{P}}\!\left[\left(q_{\rho}^{\mathbb{P}}\!\left(U\right)\right)^{\frac{\eta\,\eta'}{\eta-\eta'}}\right]^{\frac{1}{\eta'}-\frac{1}{\eta}}\mathbb{E}_{\mathbb{P}}\!\left[\left(q_{X_{n_{k}}^{-}}^{\mathbb{P}}\!\left(U\right)\right)^{\eta}\right]^{\frac{\eta'}{\eta}}\\
		&= C\,\mathbb{E}_{\mathbb{P}}\!\left[\left(X_{n_{k}}^{-}\right)^{\eta}\right]^{\frac{\eta'}{\eta}}\leq C\left(\sup_{n\in\mathbb{N}}\mathbb{E}_{\mathbb{P}}\!\left[\left(X_{n}^{-}\right)^{\eta}\right]\right)^{\frac{\eta'}{\eta}}<+\infty,
	\end{align*}
	for some $C>0$, where we use that each random variable $q_{X_{n_{k}}^{-}}^{\mathbb{P}}\!\left(U\right)$ has the same distribution as $X_{n_{k}}^{-}$, and we invoke \Cref{as:rhoinW}. Hence, by de la Vall\'{e}e-Poussin's lemma, the claim follows.
	
	The negative part function is also decreasing, so $\left[q^{\mathbb{P}}_{X_{n_{k}}}\!\left(x\right)\right]^{-}=q^{\mathbb{P}}_{X^{-}_{n_{k}}}\!\left(1-x\right)$ for a.e.\ $x\in\left(0,1\right)$ and for any $k\in\mathbb{N}$. Moreover, it is a continuous function as well, thus $\lim_{k} q^{\mathbb{P}}_{X^{-}_{n_{k}}}\!\left(x\right)=\left[q_{\nu}\!\left(1-x\right)\right]^{-}$ for a.e.\ $x\in\left(0,1\right)$. Therefore, these facts combined with uniform integrability give that
	\begin{equation*}
		\lim_{k\rightarrow+\infty} \int_{0}^{1} q_{\rho}^{\mathbb{P}}\!\left(x\right)q_{X_{n_{k}}^{-}}^{\mathbb{P}}\!\left(x\right)\,dx=\int_{0}^{1} q_{\rho}^{\mathbb{P}}\!\left(x\right)\left[q_{\nu}\!\left(1-x\right)\right]^{-}\,dx.
	\end{equation*}
	
	Consequently, it follows from the admissibility of each $X_{n_{k}}$, from the super-additivity of the $\liminf$, and from the preceding inequalities that
	\begin{align*}
		x_{0}&=\liminf_{k\rightarrow+\infty} \mathbb{E}_{\mathbb{P}}\!\left[\rho\,X_{n_{k}}\right]\\%&\geq \liminf_{k\rightarrow+\infty} \mathbb{E}_{\mathbb{P}}\!\left[\rho\,X_{n_{k}}^{+}\right]-\limsup_{k\rightarrow+\infty}\mathbb{E}_{\mathbb{P}}\!\left[\rho\,X_{n_{k}}^{-}\right]\\
		&\geq \int_{0}^{1} q_{\rho}^{\mathbb{P}}\!\left(x\right)\left[q_{\nu}\!\left(1-x\right)\right]^{+}\,dx-\lim_{k\rightarrow+\infty} \int_{0}^{1} q_{\rho}^{\mathbb{P}}\!\left(x\right)q_{X_{n_{k}}^{-}}^{\mathbb{P}}\!\left(x\right)\,dx\\
		&=\int_{0}^{1} q_{\rho}^{\mathbb{P}}\!\left(x\right)\left[q_{\nu}\!\left(1-x\right)\right]^{+}\,dx-\int_{0}^{1} q_{\rho}^{\mathbb{P}}\!\left(x\right)\left[q_{\nu}\!\left(1-x\right)\right]^{-}\,dx=\mathbb{E}_{\mathbb{Q}}\!\left[X_{*}\right],
	\end{align*}
	as intended. Finally, it is also straightforward to check that $X_{*}$ belongs to $L^{1}\!\left(\mathbb{Q}\right)$, since
	\begin{equation*}
		\mathbb{E}_{\mathbb{Q}}\!\left[\left|X_{*}\right|\right]=\mathbb{E}_{\mathbb{Q}}\!\left[X_{*}\right]+2\,\mathbb{E}_{\mathbb{Q}}\!\left[X_{*}^{-}\right]\leq x_{0}+2\,\lim_{k\in\mathbb{N}} \int_{0}^{1} q_{\rho}^{\mathbb{P}}\!\left(x\right)q_{X_{n_{k}}^{-}}^{\mathbb{P}}\!\left(x\right)\,dx<+\infty,
	\end{equation*}
	hence, by \Cref{as:kindcompl}, $X_{*}$ admits a replicating portfolio $\phi^{*}$ from initial capital $\mathbb{E}_{\mathbb{Q}}\!\left[X_{*}\right]\leq x_{0}$. A fortiori, with initial capital $x_{0}$ one also has $V(\Pi_{T}^{\phi^{*}})\geq V^{*}\!\left(x_{0}\right)$, so $\phi^{*}$ is an optimal strategy.
\end{proof}

%=========================================================================================================================================
\begin{proof}[Proof of \Cref{lem:AEzfinite}]
	Fix $\delta\in\left(0,1\right)$ arbitrary and choose $\varsigma\in\left(1,\delta^{-1/\gamma}\right)$. Then, for every $x\geq \underline{x}$, we have $z\!\left(x\right)-\delta\,z\!\left(\varsigma x\right)\geq z\!\left(x\right)\left[1-\delta\,\varsigma^{\gamma}\right]$. %
	%\begin{equation*}
	%	z\!\left(x\right)-\delta\,z\!\left(\alpha x\right)\geq z\!\left(x\right)\left[1-\delta\,\alpha^{\gamma}\right].
	%\end{equation*}
	Since $z\!\left(+\infty\right)=+\infty$ and $\delta\,\varsigma^{\gamma}<1$, we obtain that $\liminf_{x\rightarrow+\infty} \left[z\!\left(x\right)-\delta\,z\!\left(\varsigma x\right)\right]=+\infty$, and finally we use \Cref{lem:zetaG} to infer that \Cref{as:existsxi} holds true.
\end{proof}

%%%%%%%%%%%%%%%%%%%%%%%%%%%%%%%%%%%%%%%%%%%%%%%%%%%%%%%%%%%%%%%%%%%
%%                                                               %%
%% References																										 %%
%%                                                               %%
%%%%%%%%%%%%%%%%%%%%%%%%%%%%%%%%%%%%%%%%%%%%%%%%%%%%%%%%%%%%%%%%%%%

\end{document}